\documentclass{amsproc}
\usepackage{amssymb}

\theoremstyle{plain}
 \newtheorem{theorem}{Theorem}[section]
 
 \newtheorem{lemma}{Lemma}[section]
 
\theoremstyle{definition}

\theoremstyle{remark}
 \newtheorem{remark}{Remark}[section]
 \numberwithin{equation}{section}

\renewcommand{\leq}{\leqslant}
\renewcommand{\geq}{\geqslant}

\title[Asymptotic Existence of Proportionally Fair Allocations]{Asymptotic Existence of Proportionally Fair Allocations}


\author[Warut Suksompong]{Warut Suksompong}

\address{
Department of Computer Science \\ 
Stanford University   \\ 
Stanford, CA 94305\\
USA}
\email{warut@cs.stanford.edu}

\begin{document}

\vspace{18mm} \setcounter{page}{1} \thispagestyle{empty}

\begin{abstract}
Fair division has long been an important problem in the economics literature. In this note, we consider the existence of \emph{proportionally fair allocations} of indivisible goods, i.e., allocations of indivisible goods in which every agent gets at least her proportionally fair share according to her own utility function. We show that when utilities are additive and utilities for individual goods are drawn independently at random from a distribution, proportionally fair allocations exist with high probability if the number of goods is a multiple of the number of agents or if the number of goods grows asymptotically faster than the number of agents.
\end{abstract}

\maketitle

\section{Introduction}

The allocation of goods among interested agents is a task that occurs frequently in practical situations and is therefore an important issue for the society. Some goods, such as land and cake, are \emph{divisible}---each piece of land or cake can be split among multiple agents. Others, like houses and cars, are \emph{indivisible}---each house or car cannot be split among different agents. A major concern when allocating goods among agents is that the resulting allocation is \emph{fair}. Several notions of fairness have been considered in the literature. For example, an \emph{envy-free allocation} is one in which every agent values her bundle at least as much as any other agent's bundle \cite{Foley67,Varian74}, while an allocation achieving \emph{max-min fairness} is one that maximizes the minimum utility among all agents \cite{BezakovaDa05}. We refer to \cite{Brams06,Klamler10} for an overview of the fair division literature.

In this note, we consider \emph{proportionally fair allocations}, i.e., allocations in which every agent gets at least her proportionally fair share according to her own utility function \cite{Steinhaus48}. In other words, in a proportionally fair allocation, the utility that each agent gets from her allocation is at least a $1/n$ fraction of her utility of the whole set of goods. A proportionally fair allocation is the first kind of fair division studied in the literature, and is therefore sometimes referred to as a \emph{simple fair division}. The existence of a proportionally fair allocation is guaranteed if the following two conditions are satisfied simultaneously: (i) there is no indivisible good with positive value; (ii) the utility of an agent for a piece is equal to the sum of the utilities of the agent for its parts when the piece is divided into several parts. In general, however, proportionally fair allocations are not guaranteed to exist; a simple example is when we have one good and two agents both of whom value the good positively. 

Even though proportionally fair allocations of indivisible goods are not guaranteed to exist, we show that under rather general settings, they exist with high probability as the number of agents and goods grows. We assume that utilities of agents for goods are drawn independently at random from a distribution, and we make a very common assumption that utilities are \emph{additive}, i.e., the utility of an agent for a bundle of goods is the sum of the utilities of the agent for each good. We show that when the distribution does not put all probability on a single point, proportionally fair allocations exist with high probability if the number of goods is a multiple of the number of agents or if the number of goods grows asymptotically faster than the number of agents. In other words, settings in which proportionally fair allocations do not exist are rare exceptions when the number of goods and agents is large and satisfies one of the aforementioned relations.

\subsection{Related work}

Dickerson et al. \cite{DickersonGoKa14} considered asymptotic existence and nonexistence of envy-free allocations. They showed that under additive utilities, envy-free allocations are unlikely to exist even when the number of goods is larger than the number of agents by a linear fraction. On the other hand, they proved that when the number of goods is larger than the number of agents by a logarithmic factor, such allocations are likely to exist under certain technical conditions on the probability distribution. Amanatidis et al. \cite{AmanatidisMaNi15} and Kurokawa et al. \cite{KurokawaPrWa15} considered allocations that gives each agent a maximin share guarantee and showed that such allocations exist with high probability when utilities are additive. Asymptotic statements have been considered in other areas of economics as well. For instance, Manea \cite{Manea09} established that the allocation obtained by the random serial dictatorship mechanism is ordinally inefficient with high probability when agents' preference profiles are drawn at random. Incentives and stability in large matching markets have also been considered in the literature \cite{ImmorlicaMa05,KojimaPa10}.

We now make some comments on the relation between our results and those of Dickerson et al. \cite{DickersonGoKa14} concerning envy-free allocations. As we will later elaborate, proportional fairness is a weaker notion of fairness than envy-freeness, i.e., envy-free allocations are also proportionally fair when utilities are additive. Dickerson et al. showed that under additive utilities, envy-free allocations are unlikely to exist when the number of goods is larger than the number of agents by a linear fraction. Theorem \ref{smallitems} contrasts that result by showing that proportionally fair allocations are likely to exist even when the number of goods is the same as the number of agents. 

On the other hand, Dickerson et al. proved that under certain technical conditions on the probability distribution, envy-free allocations are likely to exist when the number of goods is larger than the number of agents by a logarithmic factor. (They left open the gap between constant and logarithmic factors for envy-freeness.) Theorem \ref{largeitems} shows that for proportionally fair allocations, any superconstant gap suffices to establish likelihood of existence.  Moreover, although Dickerson et al.'s result allows for the utilities of different goods to be drawn from different distributions, it does not imply our result even in the case where $m=\Omega(n\log n)$ because their result relies on technical conditions on the probability distribution, whereas Theorem \ref{largeitems} holds for general probability distributions. Remark \ref{different} reveals some limitations on generalizing our results to allow for the utilities of different goods to be drawn from different distributions. Nevertheless, we still think that this is an interesting direction that should be explored in future work.


\section{Preliminaries}

Let $N=\{1,2,\ldots,n\}$ denote the set of agents, and $G$ the set of goods with $|G|=m$. Assume that the utility $u_i(g)$ of agent $i\in N$ for good $g\in G$ lies in $[0,1]$. This constraint does not introduce a loss of generality; since we will not engage in comparisons of utilities across agents, we can scale down all utilities by their maximum. As is very common, we assume that utilities are \emph{additive}, i.e., $u_i(G')=\sum_{g\in G'}u_i(g)$ for any agent $i\in N$ and any subset of goods $G'\subseteq G$.

For agents $i\in N$ and goods $g\in G$, the utilities $u_i(g)$ are drawn independently from a distribution $\mathcal{D}$ with constant mean $\mu$. An \emph{allocation} $\mathcal{G}=(G_1,G_2,\ldots,G_n)$ is a partition of goods into bundles for the agents so that agent $i$ receives bundle $G_i$. An allocation $\mathcal{G}$ is said to be \emph{proportionally fair} if $u_i(G_i)\geq\frac{1}{n}\cdot u_i(G)$ for all $i\in N$, that is, every agent gets at least her proportionally fair share according to her own utility function. 

Note that when utilities are additive, an envy-free allocation is also proportionally fair. Indeed, if an allocation is envy-free, every agent likes her bundle at least as much as the bundle of any other agent. Additivity then implies that her bundle is worth at least a $1/n$ fraction of the whole set of goods. On the other hand, a proportionally fair allocation is always envy-free when there are two agents, but not when there are at least three agents. An example is when an agent thinks that her bundle is worth $1/3$ of the whole set of goods, while the bundle of another agent is worth the remaining $2/3$ to the first agent.

If $\mathcal{D}$ puts all probability on a single point, it is clear that a proportionally fair allocation exists if and only if the number of items $m$ is a multiple of $n$. We assume henceforth that $\mathcal{D}$ does not put all probability on a single point. Nevertheless, we will allow $\mathcal{D}$ to be a discrete distribution, an assumption that holds in many natural settings. For instance, $\mathcal{D}$ can take the value $0$ with some probability $p$ (if the agent does not like the good) and the value $1$ with the remaining probability $1-p$ (if she likes the good).

We begin by stating a general property of distributions which we will need for our results. We omit the proof since it is straightforward.

\begin{lemma}
Let $\mathcal{D}$ be a distribution with mean $\mu$ that does not put all probability on a single point, and let $X$ be a random variable drawn from $\mathcal{D}$. Then there exist constants $\delta,\beta>0$ such that $\text{Pr}[X\geq (1+\delta)\mu]\geq\beta.$
\end{lemma}


Hence from now on, we will assume the existence of constants $\delta,\beta>0$ such that $\text{Pr}[X\geq (1+\delta)\mu]\geq\beta$ for a random variable $X$ drawn from $\mathcal{D}$.

For our results, we will also need the following well-known bound. The proof can be found, for example, in \cite{MitzenmacherUp05}.

\begin{lemma}[Chernoff]
\label{chernoff}
Let $X_1,X_2,\ldots,X_m$ be independent random variables taking values in $[0,1]$, and let $X$ denote their sum. Then for any $\epsilon\in(0,1)$, we have \[\emph{Pr}[X\geq (1+\epsilon)E[X]]\leq e^{-\frac{\epsilon^2E[X]}{3}}.\]
\end{lemma}

Finally, we define asymptotic (Landau) notations used in this paper. Given two functions $f$ and $g$, we write $f(n)=\omega(g(n))$ or $g(n)=o(f(n))$ to mean that $f$ dominates $g$ asymptotically. In other words, for every fixed positive number $k$, for all sufficiently large $n$, we have $f(n)\geq kg(n)$. Similarly, we write $f(n)=\Omega(g(n))$ to mean that $f$ is bounded below by $g$ asymptotically. That is, there exists a fixed positive number $k$ such that for all sufficiently large $n$, we have $f(n)\geq kg(n)$. We refer to \cite{CormenLeRi09} for a thorough treatment of asymptotic notations.


\section{Our results}

We are now ready to state our results. For our first result, we show that proportionally fair allocations exist with high probability if the number of goods is a multiple of the number of agents.

\begin{theorem}
\label{smallitems}
Let $m=kn$ for some constant positive integer $k$. Then a proportionally fair allocation exists with probability approaching $1$ as $n\rightarrow\infty$ (or equivalently, as $m\rightarrow\infty$).
\end{theorem}

Before we go on to the formal proof of the theorem, we sketch its outline here. We first show the statement for $k=1$, i.e., the number of goods is equal to the number of agents. The high-level idea is that when this number is large, the utility that an agent has for the whole set of goods is unlikely to be much higher than the corresponding expected utility. This means that we only need to match each agent to a good that she values slightly more than average. Since each pair of good and agent satisfies this condition independently and with constant probability, such a matching is likely to exist. For the case of general $k$, we simply divide the goods into $k$ groups and perform a matching for each group.

\begin{proof}
We first show the statement for $k=1$. Consider an agent $i\in N$, and let $U_i=\sum_{g\in G} u_i(g)$. Then $E[U_i]=n\mu$. By Lemma \ref{chernoff}, we have \[\text{Pr}[U_i\geq (1+\delta)n\mu]\leq e^{-\frac{\delta^2n\mu}{3}}.\] By the union bound, with probability that goes to $1$ as $n\rightarrow\infty$, it holds that $U_i\leq (1+\delta)n\mu$ for all $i$. This implies that it suffices to match each agent with a good that she values at least $(1+\delta)\mu$. By the assumption on our distribution $\mathcal{D}$, the probability of each good satisfying this condition is at least $\beta$.

Consider a bipartite graph with one set of vertices corresponding to agents and the other set of vertices corresponding to goods. There exists an edge between an agent and a good if the agent values the good at least $(1+\delta)\mu$. Hence each edge exists with probability at least $\beta$ independently of the remaining edges. By a result on random matrices (see, e.g., Theorem 2 of \cite{ErdosRe64}), the probability that there does not exist a perfect matching is exponentially small. It follows that a proportionally fair allocation exists with probability $1$ as $n\rightarrow\infty$, as desired.

We now pass to the case of a general constant positive integer $k$. Divide the goods into $k$ groups, each with $n$ goods. For each group, we find a matching between agents and goods as before. By the union bound, the probability that there does not exist a perfect matching for some group is still exponentially low. We then assign to each agent the $k$ goods obtained from each of the groups. It is easy to see that the resulting allocation is proportionally fair.
\end{proof}

Note that the argument in the proof of Theorem \ref{smallitems} also shows the existence with high probability of a proportionally fair allocation that gives all agents an equal number of goods. As a result, we obtain a proportionally fair allocation that seems fair in the eyes of an outsider as well. Moreover, the probability that a proportionally fair allocation does not exist is exponentially small in the number of goods. This contrasts with the polynomially small probability obtained for envy-free allocations by Dickerson et al. \cite{DickersonGoKa14}

Clearly, if all agents have a positive utility for each good, the number of goods needs to be at least as large as the number of agents for a proportionally fair allocation to exist. However, the following remark shows that this condition is not sufficient for a proportionally fair allocation to exist with high probability.

\begin{remark}
\label{divisible}
Let $m=2n-1$, and suppose that utilities are drawn from the uniform distribution on $[0.4,0.6]$. Each agent's utility of the whole set of goods is at least $0.4\cdot (2n-1)$, and so in a proportionally fair allocation she needs a bundle with utility at least $0.4\cdot\frac{2n-1}{n}=0.8-\frac{0.4}{n}>0.6$ for $n\geq 3$. It follows that every agent needs a bundle with at least two goods, but this is impossible since we would need at least $2n$ goods in total.
\end{remark}

Next, we show that we cannot generalize Theorem \ref{smallitems} to the setting where the utilities of goods can be drawn from different distributions.

\begin{remark}
\label{different}
Let $m=n$, and suppose that the utilities of half of the goods are drawn from the uniform distribution on $[0,0.1]$, and the utilities of the other half from $[0.9,1]$. Each agent's utility of the whole set of goods is at least $0.9\cdot \frac{1}{2}n=0.45n$, and so in a proportionally fair allocation an agent needs a bundle with utility at least $0.45$. But this is impossible, since each agent can receive only one good and half of the goods yield utility at most $0.1$.
\end{remark}

Even though Remark \ref{divisible} shows the importance of the assumption that the number of goods is a multiple of the number of agents when the ratio between the two quantities is bounded by a constant, the significance in fact disappears when the ratio grows. In particular, the following theorem shows that proportionally fair allocations also exist with high probability if the number of goods grows asymptotically faster than the number of agents.

\begin{theorem}
\label{largeitems}
Let $m=\omega(n)$. Then a proportionally fair allocation exists with probability approaching $1$ as $n\rightarrow\infty$ (or equivalently, as $m\rightarrow\infty$).
\end{theorem}

\begin{proof}
Consider an agent $i\in N$, and let $U_i=\sum_{g\in G} u_i(g)$. Then $E[U_i]=m\mu$. By Lemma \ref{chernoff}, we have \[\text{Pr}\left[U_i\geq \left(1+\frac{\delta}{2}\right)m\mu\right]\leq e^{-\frac{\delta^2m\mu}{12}}.\] By the union bound, with probability that goes to $1$ as $n\rightarrow\infty$, it holds that $U_i\leq \left(1+\frac{\delta}{2}\right)m\mu$ for all $i$. This implies that it suffices to give each agent a set of goods that she values at least $\left(1+\frac{\delta}{2}\right)\frac{m}{n}\mu$.

Divide the goods into groups of $n$ goods, and ignore the leftover goods if there are any. For each agent and each group of goods, we will try to match each agent with a good that she values at least $(1+\delta)\mu$. As in the proof of Theorem \ref{smallitems}, the probability that there does not exist such a matching for any particular group is exponentially small.

We need to obtain a matching for at least $\frac{1+\frac{\delta}{2}}{1+\delta}\cdot\frac{m}{n}$ groups. Since we possibly throw away some leftover goods in the first step, we cannot immediately conclude that this corresponds to a constant fraction of the groups. Nevertheless, using our assumption that $n=o(m)$, it suffices to obtain a matching for at least an $\alpha$ fraction of the groups for some constant $\frac{1+\frac{\delta}{2}}{1+\delta}<\alpha<1$. Indeed, for such $\alpha$ we have 
\[\alpha\cdot\frac{m-n}{n}>\frac{1+\frac{\delta}{2}}{1+\delta}\cdot\frac{m}{n}\]
for large enough $m$, where the term $m-n$ corresponds to the fact that we throw away at most $n$ (in fact, $n-1$) items.

Since the probability that there does not exist such a matching for any particular group is exponentially small in $n$, and the events of failing are independent for distinct groups, we find that the probability that we find a matching for at least an $\alpha$ fraction of the groups approaches $1$. It follows that a proportionally fair allocation exists with probability $1$ as $n\rightarrow\infty$, as desired.
\end{proof}

\subsubsection*{Acknowledgments.} The author thanks an associate editor and the anonymous reviewers for their helpful feedback and acknowledges support from a Stanford Graduate Fellowship.


\end{document}